\documentclass {amsart}
\usepackage [utf8] {inputenc}
\usepackage {amsfonts}
\usepackage {amssymb}
\usepackage {amsmath}
 \usepackage {graphicx}

\newcommand{\PP}{\mathcal{P}}
\newcommand{\N}{\mathcal{N}}

\newtheorem{dfn}{Definition}
\newtheorem{prop}[dfn]{Proposition}
\newtheorem{theo}[dfn]{Theorem}
\newtheorem{cor}[dfn]{Corollary}

\newtheorem{rem}[dfn]{Remark}

\newtheorem{prob}[dfn]{Problem}

\newcommand{\delete}[1]{}

\title {Vertex nim played on graphs}

\author[E. Duch\^ene]{Eric Duch\^ene}
\author[G. Renault]{Gabriel Renault}

\address[E. Duch\^ene]{\newline
Universit\'e de Lyon, CNRS\newline
Universit\'e Lyon 1, LIRIS, UMR5205, F-69622, France\newline
{\tt eric.duchene@univ-lyon1.fr}
}

\address[G. Renault]{\newline
Universit\'e Bordeaux 1, LaBRI, France\newline
{\tt gabriel.renault@labri.fr}}     

\begin{document}

\maketitle

\begin{abstract}
Given a graph $G$ with positive integer weights on the vertices, and a token placed on some current vertex $u$, two players alternately remove a positive integer weight from $u$ and then move the token to a new current vertex adjacent to $u$. When the weight of a vertex is set to $0$, it is removed and its neighborhood becomes a clique. The player making the last move wins. This adaptation of Nim on graphs is called {\sc Vertexnim}, and slightly differs from the game {\sc Vertex NimG} introduced by Stockman in 2004. {\sc Vertexnim} can be played on both directed or undirected graphs. In this paper, we study the complexity of deciding whether a given game position of {\sc Vertexnim} is winning for the first or second player. In particular, we show that for undirected graphs, this problem can be solved in quadratic time. Our algorithm is also available for the game {\sc Vertex NimG}, thus improving Stockman's exptime algorithm. In the directed case, we are able to compute the winning strategy in polynomial time for several instances, including circuits or digraphs with self loops.
\end{abstract}
~

\section{Background and definitions}

We assume that the reader has some knowledge in combinatorial game theory. Basic definitions can be found in \cite{Win}. We only remind that a $\PP$ position denotes a position from which the second player has a winning strategy, while an $\N$ position means that the first player to move can win.
Graph theoretical notions used in this paper will be standard and according to \cite{BM76}. In particular, given a graph $G=(V,E)$ and a vertex $v$ of $V$, we set $N(v)=\{w\in V:(v,w)\in E\}$.  \\

The original idea of this work is the study of a variant of Nim, called {\sc Adjacent Nim}, in which both players are forced to play on the heaps in a specific cyclic order: given $N$ heaps of tokens of respective sizes $(n_1,\ldots,n_N)$, play the game of Nim under the constrainst that if your opponent has moved on heap $i$, you must move on heap $i+1$ (or on the smallest next non-empty heap, in a circular way). Actually, our investigations led us to consider {\sc Adjacent Nim} as a particular instance of the game {\sc NimG} (for "Nim on Graphs") introduced by Stockman in \cite{Sto}. \\

As a brief story of the game, we remind the reader that the game of Nim was introduced and solved by Bouton in 1904 \cite{Nim}. Since then, lots of variations were considered in the literature, the most famous one being Wythoff's game \cite{Wyt,Fra}. One can also mention \cite{DG,Fra3,Rat,Ral,Fra2} as a non-exhaustive list. One of the most recent variant of Nim provides a topology to the heaps, which are organized as the edges of an undirected graph. This game was proposed by Fukuyama in 2003 \cite{fuku,fuku2}. More precisely, an instance of its game is an undirected graph $G=(V,E)$ with an integer weight function on $E$. A token is set on an arbitrary vertex. Then two players alternately move the token along a positive adjacent edge $e$ and decrease the label of $e$ to any strictly smaller non-negative integer. The first player unable to move loses the game (this happens when the token has all its adjacent edges with a label equal to zero). In his papers, Fukuyama gives necessary and sufficient conditions for a position on a bipartite graph to be $\PP$. He also computes the Grundy values of this game for some specific families of bipartite graphs, including trees, paths or cycles. In \cite{Erik}, a larger set of graphs is investigated (including complete graphs), but only for the weight fonction $f:E\mapsto \{1\}$.  \\

In 2004, Stockman considered another generalization of Nim on graphs that she called {\sc Vertex NimG}. The main difference with Fukuyama's work is that the Nim heaps are embedded into the vertices of a graph. This definition raises a natural question when playing the game: does the player first remove some weight from a vertex and then move to another one, or does he first move to a vertex and then remove weight from it? 
\begin{itemize}
\item The variant {\it Move then remove} of {\sc Vertex NimG} was recently investigated by Burke \& George in \cite{BG}. They showed that in the case where each vertex of the input graph $G$ has a self loop, then this game is PSPACE-hard. To the best of our knowledge, nothing was proved in the general case yet.
\item The variant {\it Remove then move} of {\sc Vertex NimG} is the one that was considerd by Stockman in \cite{Sto}. In the case where the weight function is bounded by a constant, she gave a polynomial time algorithm to decide whether a given position is $\PP$ or $\N$. The same algorithm can be applied in the general case, but becomes exponential according to the order of $G$.
\end{itemize}
~\\
In Fukuyama's or Stockman's definitions, the game ends when the player is blocked because of a null weight. This means that unlike the original game of Nim, their variants may end with remaining weight on the graph. To be closer to the original Nim, we have defined the rules of our variant of {\sc Vertex NimG} in such a way that the game ends only when all the weight is removed from the graph. This variant was introduced on both directed and undirected graphs with possible loops, and under the {\it remove then move} convention. Multiple edges are not considered, since the weight is set on the vertices. We start by giving the definition of our game on undirected graphs, which is called {\sc Undirected vertexnim}.

\begin{dfn}
\label{undireted}
{\sc Undirected vertexnim}. Let $G=(V,E)$ be an undirected connected graph, let $w:V\rightarrow \mathbb{N}_{>0}$ be a function which assigns to each vertex a positive integer. Let $u\in V$ be a starting current vertex. In this game, two players alternately decrease the value of the current vertex $u$ and choose an adjacent vertex of $u$ as the new current vertex. When the value $w(v)$ of a vertex $v$ is set to $0$, then $v$ and its incident edges are removed from $G$, the subgraph $N(v)$ of $G$ becomes a clique, and a loop is added on each vertex of $N(v)$. The game ends when $G$ is empty. The player who makes the last move wins the game. 
\end{dfn}

In this definition, we make $N(v)$ become a clique after $v$ reaches zero to prevent the graph to be disconnected. In other words, we can say that in order to choose the next current vertex, it suffices to follow any path of zero vertices ending on a non zero vertex. We also add loops to prevent a player to be blocked on a vertex. The example below shows an execution of the game, the current vertex being the one with the triangle.

\begin{figure}[!ht]
	  \label{fig:exemple1}
	  \centering
		  \includegraphics[width=90mm]{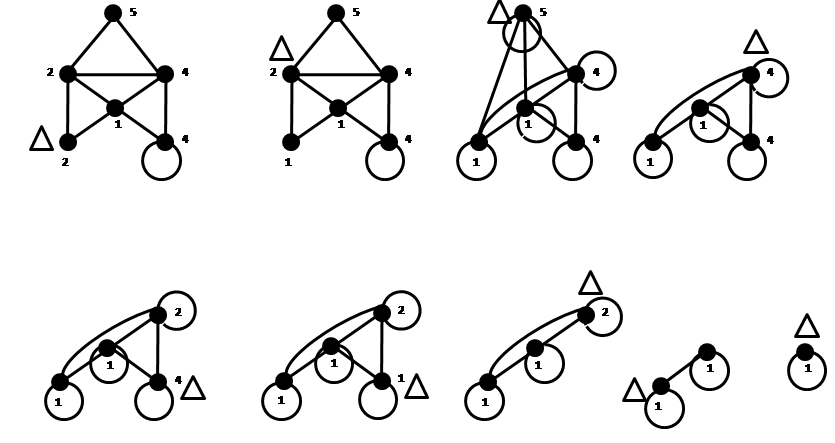}
	  \caption{Playing {\sc undirected vertexnim}}
  \end{figure}

This game can naturally be extended to directed graphs, with some constraints ensuring that all the weight is removed in the end. In particular, arcs are added when the weight of a vertex goes to zero (by the same way that a clique is build in the undirected case). We also need to play on a strong connected digraph, to avoid to be blocked on a vertex having a null outdegree. Recall that in a {\it strong connected digraph}, for every couple of vertices $(u,v)$ there exists a path from $u$ to $v$. This directed variant will be called {\sc Directed vertexnim}. 
 
\begin{dfn}
\label{directed}
{\sc Directed vertexnim}. Let $G=(V,E)$ be a strong connected digraph, and let $w:V\rightarrow \mathbb{N}_{>0}$ be a function which assigns to each vertex a positive integer. Let $u\in V$ be the starting current vertex. In this game, two players alternately decrease the value of the current vertex $u$ and choose an adjacent vertex of $u$ as the new current vertex. When the value of a vertex $v$ is set to $0$, then $v$ is removed from $G$ and all the pairs of arcs $(p,v)$ and $(v,s)$ (with $p$ and $s$ not necessarily distinct) are replaced by an arc $(p,s)$. The game ends when $G$ is empty. The player who made the last move wins the game. 
\end{dfn}

Note that in Definition \ref{directed}, the strong connectivity of $G$ is preserved when deleting a vertex. Hence it is always possible to play whenever $G$ is not empty. Figure \ref{fig:ex_directed} illustrates a sequence of moves of {\sc Directed vertexnim}, where both players remove all the weight of the current vertex at their turn.

\begin{figure}[!ht]
	  \label{fig:ex_directed}
	  \centering
		  \includegraphics[width=90mm]{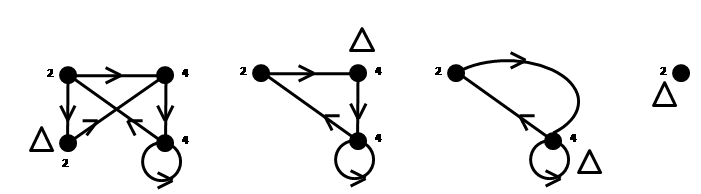}
	  \caption{Playing {\sc directed vertexnim}}
  \end{figure}
  
The current paper deals with the complexity of both versions of {\sc vertexnim}, in the sense of Fraenkel \cite{Fra02}. In particular, we will prove the tractability of the game, which implies to show that the outcome ($\PP$ or $\N$) of a game position can be computed in polynomial time. In Section 2, we will solve the game {\sc Adjacent Nim}, which is actually an instance of {\sc directed vertexnim} on circuits. Section 3 will be devoted to the resolution of {\sc directed vertexnim} for any strong connected digraph having a loop on each vertex. Section 4 concerns {\sc undirected vertexnim}, whose tractability is proved in the general case. As a corollary, we will show that our algorithm also solves Stockman's {\sc Vertex NimG} in quadratic time, improving the results presented in \cite{Sto}. In Section 5, we finally mention how our results can be adapted to mis\`ere versions of {\sc vertexnim}.

\section{Adjacent Nim}

As explained in the introduction, this game was the original motivation of our work. With the above formalism, it can be expressed as an instance of {\sc directed vertexim} on an elementary circuit $C_N=(v_1,v_2,\ldots,v_N)$ with the orientation $(v_i,v_{i+1}): 1\le i< N$ and $(v_N,v_1)$ of the arcs. In Theorem \ref{adj}, we fully solve {\sc adjacent Nim} in the case where all the weights are strictly greater than $1$. Without loss of generality, we will assume that the starting position is always $v_1$. 

\begin{theo}\label{adj}
Let $(C_N,w,v_1):N\geq 3$ be an instance of {\sc adjacent Nim} with $w:V\rightarrow \mathbb{N}_{>1}$. 
\begin{itemize}
\item If $N$ is odd, then $(C_N,w,v_1)$ is an $\N$ position. 
\item If $N$ is even, then $(C_N,w,v_1)$ is an $\N$ position iff $\min\{\underset{1\leq i \leq N}{\operatorname{argmin}}\;w(v_i)\}$ is even. 
\end{itemize}
\end{theo}

\begin{proof}
$\bullet$ If $N$ is odd, then the first player can apply the following strategy to win: first play $w(v_1)\rightarrow 1$. Then for all $1\leq i<(N-1)/2$: if the second player empties $v_{2i}$, then the first player also empties the following vertex $v_{2i+1}$. Otherwise play $w(v_{2i+1})\rightarrow 1$. The strategy is different for the last two vertices of $C_N$: if the second player empties $v_{N-1}$, then play $w(v_N)\rightarrow 1$, otherwise play $w(v_N)\rightarrow 0$. As $w(v_1)=1$, the second player is now forced to empty $v_1$. Since an even number of vertices was deleted since then, we still have an odd circuit to play on. It now suffices for the first player to empty all the vertices on the second run. Indeed, the second player is also forced to set each weight to $0$ since he has to play on vertices satisfying $w=1$. Since the circuit is odd, the first player is guaranteed to make the last move on $v_N$ or $v_{N-1}$.\\
$\bullet$ If $N$ is even, we claim that who must play the first vertex of minimum weight will lose the game. The winning strategy of the other player consists in decreasing by $1$ the weight of each vertex at his turn. Without loss of generality, assume that $\min\{\underset{1\leq i \leq N}{\operatorname{argmin}}\;w(v_i)\}$ is odd. If the strategy of the second player always consists in moving $w(v_i)\rightarrow w(v_i)-1$, then the first player will be the first to set a weight to $0$ or $1$. If he sets a vertex to $0$, then the second player now faces an instance $(C'_{N-1},w')$ with $w':V'\rightarrow \mathbb{N}_{>1}$, which is winning according to the previous item. If he sets a vertex to $1$, then the second player will empty the following vertex, leaving to the first player a position $(C'_{N-1}=(v'_1,v'_2,\ldots,v'_{N-1}),w')$ with $w':V'\rightarrow \mathbb{N}_{>1}$ except on $w'(v'_{N-1})=1$. This position corresponds to the one of the previous item after the first move, and is thus losing.
\end{proof}

\begin{prob}
The question of deciding whether a given position is $\PP$ or $\N$ remains open in the cases where some vertices have a weight equal to $1$. Indeed the previous strategy cannot be applied anymore, and we did not manage to get satisfying results when the $1's$ are owned by different players.
\end{prob}

\begin{rem}
What if we adapt Stockman's {\sc Vertex NimG} to directed graphs ? Recall that it means that vertices of null weight are never removed, and a player who must play from a $0$ loses. In the case of circuits, it is easy to see that Theorem \ref{adj} remains true, even if there are vertices of weight $1$. On a general graph, we conjecture that this game should be at least as hard as the game {\sc Geography} \cite{FrSi} (nevertheless the reduction needs to be done).
\end{rem}

\section{Directed graphs with all loops}

Dealing with {\sc directed vertexnim} on any strong connected digraph is much harder. We managed to decide whether a position is $\PP$ or $\N$ only in the case where there is a loop on each vertex. This is somehow a way to consider {\sc NimG} with the extended neighborhood, as proposed in \cite{BG}. 

\begin{theo}
\label{thm:dg}
Let $(G,w,u)$ be an instance of {\sc directed vertexnim} where $G$ is strongly connected with a loop on each vertex.
Deciding whether (G,w,u) is $\PP$ or $\N$ can be done in time $O(|V(G)||E(G)|)$.
\end{theo}

The proof of this theorem requires several definitions that we present here.

\begin{dfn}
\label{def:thmdir}
Let $G=(V,E)$ be a directed graph.
We define a labeling $lo_G:V(G)\rightarrow\{\PP,\N\}$ as follows : \\
Let $S \subseteq V(G)$ be a non-empty set of vertices such that the graph induced by $S$ is strongly connected and $\forall u \in S, \forall v \in (V(G) \backslash S)$, $(u,v) \notin E(G)$. \\
Let $T = \{v \in V(G) \backslash S \mid \exists u \in S, (v,u) \in E(G)\}$. \\
Let $G_e$ be the graph induced by $V(G) \backslash S$ and $G_o$ the graph induced by $V(G) \backslash (S \cup T)$. \\
If $|S|$ is even, $\forall u \in S$, $lo_G(u) = \N$, and $\forall v \in G \backslash S$, $lo_G(v) = lo_{G_e}(v)$. \\
If $|S|$ is odd, $\forall u \in S$, $lo_G(u) = \PP$, $\forall v \in T$, $lo_G(v) = \N$ and $\forall w \in G \backslash (S \cup T)$, $lo_G(w) = lo_{G_o}(w)$.
\end{dfn}
When decomposing the graph into strongly connected components, $S$ is one of those with no out-arc.
The choice of $S$ is not unique, unlike the $lo_{G}$ function: if $S_1$ and $S_2$ are both strongly connected components without out-arcs, the one which is not chosen as first will remain a strongly connected component after the removal of the other, and as it has no out-arc, none of its vertices will be in the $T$ set.

\begin{proof}
Let $G'$ be the induced subgraph of $G$ such that $V(G') = \{v \in V(G) \mid w(v)=1\}$. \\
If $G=G'$, then $(G,w,u)$ is an $\N$ position if and only if $|V(G)|$ is odd since the problem reduces to \textquotedblleft She loves move, she loves me not\textquotedblright. We will now suppose that $G \neq G'$, and consider two cases about $w(u)$: \\
~\textbullet~ Assume $w(u) \geqslant 2$. 
If there is a winning move which reduces $u$ to $0$, then we can play it and win.
Otherwise, reducing $u$ to $1$ and staying on $u$ is a winning move.
Hence $(G,w,u)$ is an $\N$ position. \\
~\textbullet~ Assume $w(u) = 1$, i.e., $u\in G'$.
According to Definition \ref{def:thmdir}, computing $lo_{G'}$ yields a sequence of couples of sets $(S_i,T_i)$ (which is not unique). Note that some $T_i$ may be empty (this happens when the corresponding $S_i$ has an even size). Thus the following assertions hold: if $u \in S_i$ for some $i$, then any direct successor $v$ of $u$ is in a set $S_j$ or $T_k$ with $j \leqslant i$ and $k < i$, and if $u \in T_i\neq \emptyset$ for some $i$, then there exists a direct successor $v$ of $u$ in the set $S_i$, with $lo_{G'}(v)= \PP$. \\
Our goal is to show that $(G,w,u)$ is an $\N$ position if and only if $lo_{G'}(u) = \N$ by induction on $|V(G')|$.
If $|V(G')| = 1$, then $V(G') = \{u\}$ and $lo_{G'}(u) = \PP$. Hence we are forced to reduce $u$ to $0$ and go to a vertex $v$ such that $w(v) \geqslant 2$, which we previously proved to be a losing move.
Assume $|V(G')| \geqslant 2$.
First, note that when one reduce the weight of a vertex $v$ to $0$, the replacement of the arcs makes the strongly connected components remain the same (except the component containing $v$ of course, which loses one vertex).
Consequently, if $u \in S_i$ for some $i$, then for any vertex $v \in \cup_{l=1}^{i-1} (T_l \cup S_l)$, $lo_{G'\backslash \{u\}}(v) = lo_{G'}(v)$ and for any vertex $w \in S_i \backslash \{u\}$, $lo_{G'\backslash \{u\}}(w) \neq lo_{G'}(w)$. If $u \in T_i$ for some $i$, then for any vertex $v \in (\cup_{l=1}^{i-1} (T_l \cup S_l)) \cup S_i$, $lo_{G'\backslash \{u\}}(v) = lo_{G'}(v)$. \\
We now consider two cases about $u$: first assume that $lo_{G'}(u) = \PP$, with $u \in S_i$ for some $i$.
We reduce $u$ to $0$ and we are forced to move to a direct successor $v$.
If $w(v) \geqslant 2$, we previously proved this is a losing move.
If $v \in \cup_{l=1}^{i-1} (T_l \cup S_l)$, then $lo_{G'\backslash \{u\}}(v) = lo_{G'}(v) = \N$ and it is a losing move by induction hypothesis.
If $v \in S_i$, then $lo_{G'\backslash \{u\}}(v) \neq lo_{G'}(v) = \PP$ and it is a losing move by induction hypothesis. \\
Now assume that $lo_{G'}(u) = \N$.
If $u \in T_i$ for some $i$, we can reduce $u$ to $0$ and move to a vertex $v \in S_i$, which is a winning move by induction hypothesis.
If $u \in S_i$ for some $i$, it means that $|S_i|$ is even, we can reduce $u$ to $0$ and move to a vertex $v \in S_i$, with $lo_{G' \backslash \{u\}}(v) \neq lo_{G'}(v) = \N$. This is a winning move by induction hypothesis.
Hence, $(G,w,u)$ is an $\N$ position if and only if $lo_{G'}(u) = \N$.
Figure \ref{fig:dir3} illustrates the computation of the $lo$ function.
\end{proof}

\begin{figure}[!ht]
\centering
	  \label{fig:dir3}
		  \includegraphics[scale=.5]{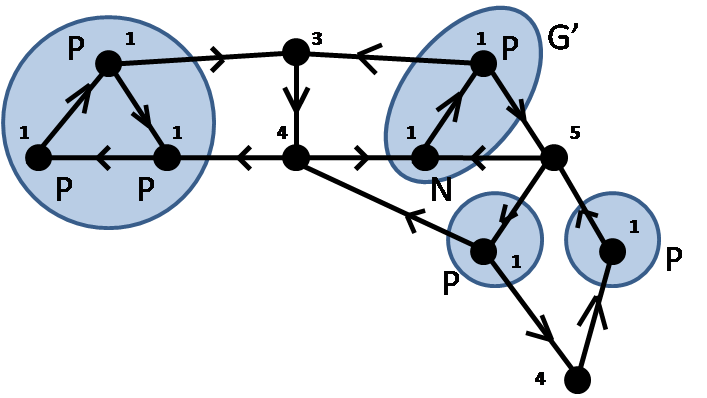}
	  \caption{Exemple of $lo$ labeling function}
\end{figure}

\begin{prob}
Can one provide a characterization of the $\PP$ and $\N$ positions in the general case where self loops are optional? \\
Note that one of the reasons for which we have slightly changed Stockman's rules is that we assumed that our current game had a lower complexity than {\sc Vertex NimG} or {\sc Geography} \cite{FrSi} on directed graphs. The previous theorem shows that our assumption was somehow true, since we remind the reader that {\sc Vertex NimG} was proved to be PSPACE-hard with all loops (for the {\it move then remove} convention \cite{BG}, the other convention being trivial with all loops). 
\end{prob}

\section{Undirected graphs}

In the undirected case, it is easy to show that if each vertex has a self loop, deciding whether a position is $\PP$ or not only depends on the size of the subset $\{v \in V \mid w(v)=1\}$. Remark that this game can be solved by Theorem \ref{thm:dg}, by saying that it suffices to replace each edge $(u,v)$ by two arcs $(u,v)$ and $(v,u)$. Yet, the following proposition improves the complexity of the method, which becomes linear.

\begin{prop}
Let $(G=(V,E),w,u)$ be an instance of {\sc undirected vertexnim} such that there is a loop on each vertex of $G$.
Deciding whether $(G,w,u)$ is $\PP$ or $\N$ can be done in time $O(|V|)$.
\end{prop}

\begin{proof}
Let $G'$ be the induced subgraph of $G$ such that $V(G') = \{v \in V(G) \mid w(v)=1\}$. \\
If $G = G'$, then $(G,w,u)$ is an $\N$ position if and only if $|V(G)|$ is odd since the problem reduces to \textquotedblleft She loves move, she loves me not\textquotedblright. In the rest of the proof, assume $G \neq G'$. \\
~\textbullet~ We first consider the case where $w(u) \geqslant 2$. 
If there is a winning move which reduces $u$ to $0$, then we play it and win.
Otherwise, reducing $u$ to $1$ and staying on $u$ is a winning move.
Hence $(G,w,u)$ is an $\N$ position. \\
~\textbullet~ Assume $w(u) = 1$.
Let $n_u$ be the number of vertices of the connected component of $G'$ which contains $u$.
We show that $(G,w,u)$ is an $\N$ position if and only if $n_u$ is even by induction on $n_u$.
If $n_u = 1$, then we are forced to reduce $u$ to $0$ and move to another vertex $v$ having $w(v) \geqslant 2$, which we previously proved to be a losing move. 
Now assume $n_u \geqslant 2$. If $n_u$ is even, we reduce $u$ to $0$ and move to an adjacent vertex $v$ with $w(v) = 1$, which is a winning move by induction hypothesis. 
If $n_u$ is odd, then we reduce $u$ to $0$ and we are forced to move to an adjacent vertex $v$. If $w(v) \geqslant 2$, then we previously proved it is a losing move. If $w(v) = 1$, this is also a losing move by induction hypothesis. Therefore in that case, $(G,w,u)$ is an $\N$ position if and only if $n_u$ is even.
\end{proof}

In the general case where the loops are optional, the tractability of the game is still guaranteed, even though the previous linear time algorithm is no more available. 

\begin{theo}
\label{thm:ug}
Let $(G,w,u)$ be an instance of {\sc undirected vertexnim}.
Deciding whether $(G,w,u)$ is $\PP$ or $\N$ can be done in $O(|V(G)||E(G)|)$ time.
\end{theo}

The proof of this theorem requires several definitions that we present here.

\begin{dfn}
Let $G=(V,E)$ be an undirected graph with a weight function $w:V\rightarrow \mathbb{N}_{>0}$ defined on its vertices.\\
Let $S = \{u \in V(G) \mid \forall v \in N(u), w(u) \leqslant w(v)\}$. \\
Let $T = \{v \in V(G) \backslash S \mid \exists u \in S, (v,u) \in E(G)\}$. \\
Let $G'$ be the graph induced by $G \backslash (S \cup T)$. \\
We define a labeling $lu_{G,w}$ of its vertices as follows : \\
$\forall u \in S$, $lu_{G,w}(u) = \PP$, $\forall v \in T$, $lu_{G,w}(v) = \N$ and $\forall t \in G \backslash (S \cup T)$, $lu_{G,w}(t) = lu_{G',w}(t)$.
\end{dfn}

\begin{proof}
Let $G_u$ be the induced subgraph of $G$ such that $V(G_u) = \{v \in V(G) \mid w(v)=1~or~v=u\}$, and $G'$ be the induced subgraph of $G$ such that $V(G') = \{v \in V(G) \mid w(v)\geqslant 2~and~(v,v)\notin E(G)~and~\forall t \in V(G), (v,t) \in E(G) \Rightarrow w(t) \geqslant 2 \}$. \\
If $G=G_u$ and $w(u) = 1$, then $(G,w,u)$ is an $\N$ position if and only if $|V(G)|$ is odd since it reduces to \textquotedblleft She loves move, she loves me not\textquotedblright. \\
If $G=G_u$ and $w(u) \geqslant 2$, we reduce $u$ to $0$ and move to any vertex if $|V(G)|$ is odd, and we reduce $u$ to $1$ and move to any vertex if $|V(G)|$ is even; both are winning moves, hence $(G,w,u)$ is an $\N$ position. \\
In the rest of the proof we will assume that $G \neq G_u$. In the next three cases, we consider the case $u\notin G'$. \\
~\textbullet~ {\it Case (1)} Assume $w(u) \geqslant 2$ and there is a loop on $u$.
If there is a winning move which reduces $u$ to $0$, then we can play it and win.
Otherwise, reducing $u$ to $1$ and staying on $u$ is a winning move.
Therefore $(G,w,u)$ is an $\N$ position. \\
~\textbullet~ {\it Case (2)} Assume $w(u) = 1$. \\
Let $n$ be the number of vertices of the connected component of $G_u$ which contains $u$.
We will show that $(G,w,u)$ is an $\N$ position if and only if $n$ is even by induction on $n$.
If $n = 1$, then we are forced to reduce $u$ to $0$ and move to another vertex $v$, with $w(v) \geqslant 2$, which was proved to be a losing move since it creates a loop on $v$.
Now assume $n \geqslant 2$.
If $n$ is even, we reduce $u$ to $0$ and move to a vertex $v$ satisfying $w(v) = 1$, which is a winning move by induction hypothesis (the connected component of $G_u$ containing $u$ being unchanged, except the removal of $u$).
If $n$ is odd, we reduce $u$ to $0$ and move to some vertex $v$, creating a loop on it.
If $w(v) \geqslant 2$, we already proved this is a losing move.
If $w(v) = 1$, it is a losing move by induction hypothesis.
We can therefore conclude that $(G,w,u)$ is an $\N$ position if and only if $n$ is even. Figure \ref{fig:case2} 	illustrates this case.\\
~\textbullet~ {\it Case (3)} Assume $w(u) \geqslant 2$ and there is a vertex $v$ such that $(u,v) \in E(G)$ and $w(v) = 1$.
Let $n$ be the number of vertices of the connected component of $G_u$ which contains $u$.
If $n$ is odd, we reduce $u$ to $1$ and we move to $v$, which we proved to be a winning move.
If $n$ is even, we reduce $u$ to $0$ and we move to $v$, which we also proved to be winning.
Hence $(G,w,u)$ is an $\N$ position in that case. Figure \ref{fig:case3} illustrates this case.\\
\begin{figure}[h!]
   \begin{minipage}[c]{.48\linewidth}
   \begin{center}
	\includegraphics[scale=.4]{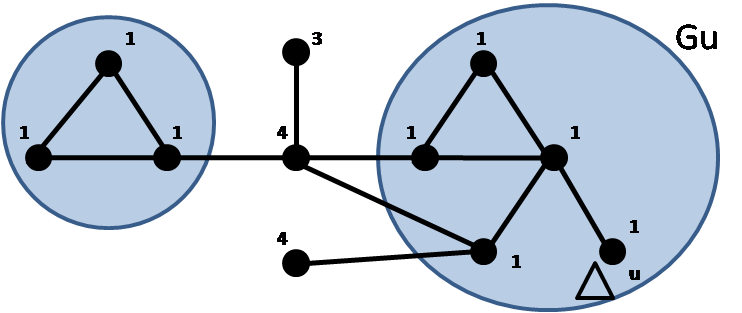}
	\caption{Case 2: the connected component containing u has an odd size: this is a $\PP$ position.}
	\label{fig:case2}
	\end{center}
   \end{minipage} \hfill
   \begin{minipage}[c]{.48\linewidth}
   \begin{center}
    \includegraphics[scale=.4]{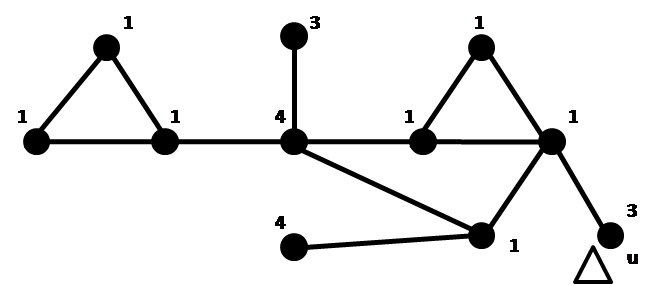}
	\label{fig:case3}
	\caption{Case 3: an $\N$ position since $u$ of weight $w(u)>1$ has a neighbor of weight $1$.}
	\end{center}
   \end{minipage}
\end{figure}
\noindent
~\textbullet~ {\it Case (4)} Assume $u \in G'$.
We will show that $(G,w,u)$ is $\N$ if and only if $lu_{G',w}(u) = \N$ by induction on $\sum_{v \in V(G')} w(v)$.
If $\sum_{v \in V(G')} w(v) = 2$, we get $G' = \{u\}$ and we are forced to play to a vertex $v$ such that $w(v) \geqslant 2$ and $v \notin V(G')$, which we proved to be a losing move.
Assume $\sum_{v \in V(G')} w(v) \geqslant 2$.
If $lu_{G',w}(u) = \N$, we reduce $u$ to $w(u)-1$ and move to a vertex $v$ of $G'$ such that $w(v) < w(u)$ and $lu_{G',w}(v) = \PP$. Such a vertex exists by definition of $lu$.
Let $(G_1,w_1,v)$ be the resulting position after such a move.
Hence $lu_{G'_1,w_1}(v) = lu_{G',w}(v) = \PP$ since the only weight that has been reduced remains greater or equal to the one of $v$. And $(G_1,w_1,v)$ is a $\PP$ position by induction hypothesis.
If $lu_{G',w}(u) = \PP$, the first player is forced to reduce $u$ and to move to some vertex $v$.
Let $(G_1,w_1,v)$ be the resulting position.
First remark that $w_1(v) \geqslant 2$ since $u\in G'$.
If he reduces $u$ to $0$, he will lose since $v$ now has a self loop.
If he reduces $u$ to $1$, he will also lose since $(u,v) \in E(G_1)$ and $w_1(u) = 1$ (according to case (3)).

Assume we reduced $u$ to a number $w_1(u) \geqslant 2$.
Thus $lu_{G'_1,w_1}(u)$ still equals $\PP$ since the only weight we modified is the one of $u$ and it has been decreased.
If $v \notin G'$, i.e., $v$ has a loop or $\exists t \in V(G_1)$ s.t. $(v,t) \in E(G_1)$ and $w_1(t) = 1$, then the second player wins according to cases (1) and (3).
If $v \in G'$ and $lu_{G',w}(v) = \N$, then $lu_{G'_1,w_1}(v)$ is still $\N$ since the only weight we modified is the one of a vertex labeled $\PP$. Consequently the resulting position makes the second player win by induction hypothesis.
If $v \in G'$ and $lu_{G',w}(v) = \PP$, then we necessarily have $w(v) = w(u)$ in $G'$. As $lu_{G'_1,w_1}(u) = \PP$ and $(u,v) \in E(G_1)$, then $lu_{G'_1,w_1}(v)$ becomes $\N$, implying that the second player wins by induction hypothesis.
Hence $(G,w,u)$ is $\N$ if and only if $lu_{G',w}(u) = \N$. Figure \ref{fig:case4} shows an example of the $lu$ labeling.\\

Concerning the complexity of the computation, note that all the cases except (4) can be executed in $O(|E(G)|)$ operations. Hence the computation of $lu_{G',w}(u)$ to solve case (4) becomes crucial. It is rather straightforward to see that in the worst case, the computation of $S$ and $T$ can be done in $O(|E(G)|)$ time. And the number of times where $S$ and $T$ are computed in the recursive definition of $lu$ is clearly bounded by $|V(G)|$. All of this leads to a global algorithm running in $O(|V(G)||E(G)|)$ time.
\begin{figure}[!ht]
\centering
	  \label{fig:case4}
		  \includegraphics[scale=.5]{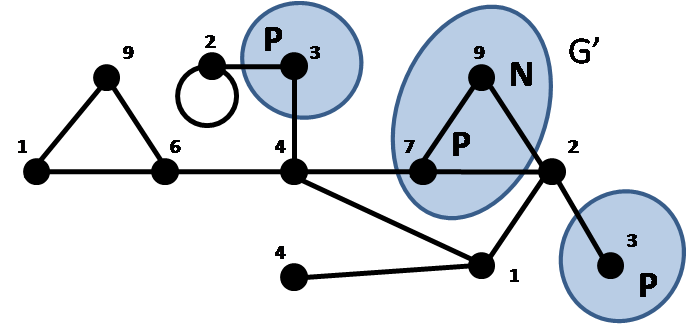}
	  \caption{Case 4: $lu$-labeling of the subgraph $G'$}
\end{figure}
\end{proof}

The technique described above can also be applied to Stockman's version of the game {\sc Vertex NimG}. In \cite{Sto}, an exptime algorithm is given to decide the outcome of a given position. We here show that the complexity can be decreased to $O(|V||E|)$. 
\begin{cor}
Let $(G,w,u)$ be an instance of {\sc Vertex NimG} with $w:V\rightarrow \mathbb{N}_{>0}$.
Deciding whether $(G,w,u)$ is $\PP$ or $\N$ can be done in $O(|V(G)||E(G)|)$ time.
\end{cor}

\begin{proof}
The proof works similarly to the previous one, except that the subgraph $G_u$ is no more useful. Hence we have four cases:
\begin{itemize}
\item If $w(u)=1$ and $u$ has no self loop, then the position is $\PP$.
\item If $w(u)\geq 1$ and there is a loop on $u$, then it is $\N$.
\item If $w(u)\geq 2$ and there is a vertex $v$ such that $(u,v)\in E$ and $w(v)=1$, then it is an $\N$ position.
\item If $u\in G'$, then compute $lu_{G',w}(u)$ as in Theorem \ref{thm:ug}.
\end{itemize}

Note that the proof is still working if there exist vertices of null weight at the beginning. It suffices to consider the two following properties: if $w(u)=0$ then this is $\PP$, and if $u$ is adjacent to some $v$ with $w(v)=0$, then it is $\N$. 
\end{proof}

\section{Mis\`ere versions}

The mis\`ere version of a game is a game with the same rules except that the winning condition is reversed, i.e., the last player to move loses the game. The following results shows that in almost all cases, mis\`ere and normal versions of {\sc Vertexnim} have the same outcomes.

\begin{theo}
Let $(G,w,u)$ be an instance of {\sc undirected vertexnim} under the mis\`ere convention.
Deciding whether $(G,w,u)$ is $\PP$ or $\N$ can be done in $O(|V(G)||E(G)|)$ time.
\end{theo}

\begin{proof}
If all vertices have weight $1$, then $(G,w,u)$ is an $\N$ position if and only if $|V(G)|$ is even since it reduces to the mis\`ere version of \textquotedblleft She loves move, she loves me not\textquotedblright.
Otherwise, we can use the same proof as the one of Theorem \ref{thm:ug} to see that $(G,w,u)$ is $\N$ in the mis\`ere version if and only if it is $\N$ in the normal version.
\end{proof}

\begin{theo}
Let $(G,w,u)$ be an instance of {\sc directed vertexnim} in the mis\`ere version, where $G$ is strongly connected, with a loop on each vertex.
Deciding whether (G,w,u) is $\PP$ or $\N$ can be done in time $O(|V(G)||E(G)|)$.
\end{theo}

\begin{proof}
If all vertices have weight $1$, then $(G,w,u)$ is an $\N$ position if and only if $|V(G)|$ is even since it reduces to the mis\`ere version of \textquotedblleft She loves move, she loves me not\textquotedblright.
Otherwise, we can use the same proof as the one of Theorem \ref{thm:dg} to see that $(G,w,u)$ is $\N$ in the mis\`ere version if and only if it is $\N$ in the normal version.
\end{proof}

\begin{rem}
Though the algorithms we give for both {\sc Undirected Vertexnim} and {\sc Directed Vertexnim} can easily be adapted for the mis\`ere version, it does not seem to be the case with the algorithm we give for {\sc Vertex NimG}.
\end{rem}

\begin{prob}
This section showed that {\sc Vertex NimG} and {\sc Undirected Vertexnim} can both be solved in polynomial time. Does this remain true when considering the {\it Move then remove} convention ?
\end{prob}

\section*{Conclusion}

When dealing with an undirected graph, we proved that both versions of {\sc Vertex NimG} (Stockman's version where the game ends whenever a player is blocked on a $0$, and our version which allows to play until all the weight is removed from the graph) are tractable. We even proved that deciding whether a given position is $\PP$ or $\N$ can be done in quadratic time, which is a real improvement compared to the exptime algorithm presented in \cite{Sto}.
Unfortunately, the directed case turns out to be more tricky, even for simple graphs such as circuits. Yet, it seems that our variant of Nim on graphs is more accessible than {\sc Vertex NimG} or {\sc Geography}, as the results obtained in Theorem \ref{thm:dg} allow us to be optimistic for graphs where loops become optional.

\end{document}